\documentclass[letterpaper,10pt,conference]{ieeeconf}  % Comment this line out
\IEEEoverridecommandlockouts                              % This command is only
\usepackage{graphics} % for pdf, bitmapped graphics files
\usepackage{amsmath} % assumes amsmath package installed
\usepackage{amssymb}  % assumes amsmath package installed
\usepackage{amsfonts}
\usepackage[english]{babel}
\usepackage[T1]{fontenc}
\usepackage{comment}
\usepackage{mathtools}
\usepackage{multirow}
\usepackage{cite}
\usepackage{tikz}
\usepackage{booktabs} 
\usepackage{tabularx}
\usepackage{url}

\newcommand{\dd}{{\mathsf{d}}}

\newcommand{\R}{\ensuremath{{\mathbb R}}}
\newcommand{\Z}{\ensuremath{{\mathbb Z}}}
\newcommand{\N}{\ensuremath{{\mathbb N}}}

\newcommand{\UU}{{\mathcal U}}

\newcommand{\XX}{{\mathcal X}}

\colorlet{Darkred}{red!50!black}
\colorlet{Darkgreen}{green!50!black}
\colorlet{Darkblue}{blue!70!black}
\colorlet{Darkorange}{orange!80!black}
\colorlet{Darkpurple}{purple!75!gray}

\newtheorem{lem}{{Lemma}}

\newtheorem{thm}{{Theorem}}
\newtheorem{asm}{{Assumption}}

\newtheorem{exam}{{Example}}

\title{\LARGE \bf Model Reference Gaussian Process Regression:\\ Data-Driven State Feedback Controller}

\author{Hyuntae Kim, Hamin Chang, and Hyungbo Shim% <-this % stops a space
\thanks{This work was supported by the grant from Hyundai Motor Company's R\&D Division.}% <-this % stops a space
\thanks{All authors are with ASRI, Department of Electrical and Computer Engineering,
Seoul National University, 1 Gwanak-ro, Gwanak-gu, Seoul, 08826, Korea. Corresponding author: {\tt\small hshim@snu.ac.kr}}%
}

\begin{document}

\maketitle
\thispagestyle{plain}
\pagestyle{plain}

%%%%%%%%%%%%%%%%%%%%%%%%%%%%%%%%%%%%%%%%%%%%%%%%%%%%%%%%%%%%%%%%%%%%%%%%%%%%%%%%
\begin{abstract}

This paper proposes a data-driven state feedback controller that enables reference tracking for nonlinear discrete-time systems. The controller is designed based on the identified inverse model of the system and a given reference model, assuming that the identification of the inverse model is carried out using only the system's state/input measurements. When its results are provided, we present conditions that guarantee a certain level of reference tracking performance, regardless of the identification method employed for the inverse model. Specifically, when Gaussian process regression (GPR) is used as the identification method, we propose sufficient conditions for the required data by applying some lemmas related to identification errors to the aforementioned conditions, ensuring that the Model reference-GPR (MR-GPR) controller can guarantee a certain level of reference tracking performance. Finally, an example is provided to demonstrate the effectiveness of the MR-GPR controller.

\end{abstract}

%%%%%%%%%%%%%%%%%%%%%%%%%%%%%%%%%%%%%%%%%%%%%%%%%%%%%%%%%%%%%%%%%%%%%%%%%%%%%%%%
\section{Introduction}

Gaussian process regression (GPR) \cite{GPML} is a nonparametric regression technique commonly used in machine learning and robotics, owing to its ability to handle large and complex datasets, integrate prior knowledge, and offer a probabilistic understanding of uncertainty, making it a powerful and versatile tool for regression analysis adopted in both academia and industry \cite{KocijanBook}.

The use of GPR for identifying unknown nonlinear systems using input/output or state data and then designing a model-based controller for the identified model has been extensively studied \cite{RDCA03, JRCA04,GEF16,CATJ16,LJM19,JLAM19,Scaramuzza}. Furthermore, the application of GPR for inverse model identification has emerged as a promising approach to controlling complex dynamic systems, such as robotics and mechatronic systems. The concept behind using GPR for inverse model identification is to utilize the inverse model that leads from the current state to the desired next state to learn the control input through GPR.
For instance, \cite{KocijanBook} and \cite{C05} present a detailed review of controlling robotic systems using inverse model GPR.
The practical efficacy of using GPR for inverse model identification in achieving high precision control has been demonstrated in several studies \cite{NSP08, NSP09, NPSS08} conducted on robot arms.
Additionally, \cite{N01} uses neural networks for inverse model-based control.

The primary emphasis of earlier studies has revolved around attaining high precision in experimental results and computational speed in real-time. Nevertheless, despite the practical effectiveness of these approaches, there is still a lack of control theoretical understanding of their fundamental principles. This theoretical gap may result in the absence of a stability guarantee within the system, thereby creating the possibility of hazardous scenarios in real-world systems. As such, it is imperative to address this issue to enhance the reliability and safety of control systems.

In this paper, we propose the data-driven inverse model-based controller as a state feedback control by using only state/input measurements of the plant and identifying the inverse model using GPR under the given reference model. We refer to this as model reference Gaussian process regression (MR-GPR) control as in \cite{HHH22}, where tracking control was not studied while the output feedback version of MR-GPR controller is firstly proposed. 

The structure of this paper is as follows. 
In Section \ref{sec:statefb}, the stability of the closed-loop system is concerned by employing bounds on the identification performance necessary for tracking control using arbitrary data-driven methods. 
Section \ref{sec:main} introduces GPR and several lemmas regarding the posterior variance of the GPR, which can be utilized to measure the regression error as confidence information about the regression result.  
Finally, based on the proposed lemmas, a stability analysis of the closed-loop system with MR-GPR is given. 
In Section \ref{sec:ex}, we offer an example to demonstrate the effectiveness of the MR-GPR controller. 
 Lastly, Section \ref{sec:conc} summarizes the paper.
 
{\it Notation:}
For integers $n, m,$ and $k$, let $0_{n\times m}\in\mathbb{R}^{n\times m}$ and $I_k\in\mathbb{R}^{k\times k}$ be the zero matrix and the identity matrix, respectively.
For column vectors $a$ and $b$,  $[a;b]$ denotes $[a^T,b^T]^T$. 
For a set $\mathcal{A}$, we define the number of elements in the set $\mathcal{A}$ as $|\mathcal{A}|.$
For discrete-time vector sequences $y(t)$ and $z(t)$,
we define a set
\begin{align*}
    &\{(y(t),z(t))\}^{k+T}_{t=k}\\
    &\quad\quad:= \{(y(k),z(k)),\cdots,(y(k+T),z(k+T))\}.
\end{align*}

%%%%%%%%%%%%%%%%%%%%%%%%%%%%%%%%%%%%%%%%%%%%%%%%%%%%%%%%%%%%%%%%%%%%%%%%%%%%%%%%
\section{System Description and Data-Driven Inverse Model-Based Control}\label{sec:statefb}

Consider a nonlinear discrete-time system
\begin{align}\label{SYSTEM:full}
x(t+1) &= f(x(t),u(t)),
\end{align}
where $t \in \Z$ is the discrete time index, $x(t)\in\R^{n}$ is the state, $u(t) \in \UU \subset \R^m$ is the input with a compact set $\UU$, and $f:\R^n \times \R^m \to \R^n$ is smooth.
Suppose that a reference model is given by
\begin{equation}\label{eq:real_closedfull}
x(t+1) = f_R(x(t),t),
\end{equation}
where $f_R: \R^n \times \Z \to \R^n$ is smooth in its first argument and uniformly bounded in the second argument.
Our goal is to construct a (data-driven) state feedback controller with which the following holds for a given error bound $\epsilon > 0$:
\begin{equation}\label{eq:closeness}
    \|x(t) - x_R(t)\| < \epsilon, \qquad \forall t \ge 0,
\end{equation}
where $x_R$ is the solution to \eqref{eq:real_closedfull} with $x_R(0)=x(0)$.
Note that some external inputs to the reference model are allowed by the time index in the function $f_R$ (see Example \ref{exam:1}).
The reference model has our desired stability, performance, robustness and so on, and we suppose that there is a compact and convex set $\XX_R \subset \R^n$, which is an operation region of \eqref{eq:real_closedfull}, such that $x_R(t) \in \XX_R$ for all $t \ge 0$.
In addition, we particularly assume that the reference model is contractive in the sense of \cite{Lohmiller98} as follows.

\begin{asm}\label{asm:contractive}
There exist a positive definite matrix $\Theta \in \R^{n \times n}$ and a scalar $\gamma<1$ such that
\begin{equation}
    \left( \frac{\partial f_R}{\partial x}(x,t) \right)^T  \Theta \left( \frac{\partial f_R}{\partial x}(x,t) \right) \le \gamma \Theta, \qquad \forall x \in \XX, t \ge 0
\end{equation}
where $\XX := \{ x \in \R^n : \exists z \in \XX_R \text{ such that } \|x-z\| \le \epsilon \}$.
\end{asm}

If the reference model is a linear system, then Assumption~\ref{asm:contractive} is nothing but the stability of the linear system.
Now, to enable model reference control, we impose the following assumption.

\begin{asm}\label{ASM:new} 
System \eqref{SYSTEM:full} satisfies the following.
\begin{itemize}
    \item[(a)] For $x(t),x(t+1)\in\R^n$ obtained from system \eqref{SYSTEM:full}, an input $u(t)\in\R^m$ that satisfies $x(t+1)=f(x(t),u(t))$ is unique.
    In other words, there exists a function
    $c: \R^{2n} \to \R^m$ such that 
    \begin{equation}
    x(t+1)=f(x(t),c([x(t);x(t+1)]))
    \end{equation}
    for all $x(t),x(t+1)$, and $t \ge 0$.
    \item[(b)] For each $x \in \XX$ and $t \ge 0$, there exists $u \in \R^m$ such that 
\begin{equation}\label{eq:asm1}
f_R(x,t) = f(x,u).
\end{equation}
\end{itemize}
\end{asm}

If system \eqref{SYSTEM:full} is input-affine, i.e., given by
$$x(t+1) = f(x(t),u(t)) = a(x(t)) + b(x(t)) u(t),$$ where $a:\R^n \to \R^n$ and $b:\R^n \to \R^{n \times m}$ are smooth, and if $b(x)$ has full column rank for all $x \in \R^n$, then Assumption \ref{ASM:new}-(a) holds with 
the smooth function $c$ given by
\begin{equation}\label{eq:c}
	c([x(t);x(t+1)]) = b^\dagger(x(t)) (x(t+1) - a(x(t))),
\end{equation}
where $b^\dagger(x(t))$ is a left-inverse of $b(x(t))$.
In addition, if the reference model $f_R$ is chosen to satisfy
\begin{equation}\label{eq:image}
	f_R(x,t) - a(x) \in {\rm Im}( b(x) ), \quad \forall x \in \XX, t \ge 0,
\end{equation}
where ${\rm Im}(b(x))$ represents the image of the matrix $b(x)$, then Assumption \ref{ASM:new}-(b) also holds.

It is trivial that under Assumption \ref{ASM:new}, the closed-loop system \eqref{SYSTEM:full} with a state feedback control input
\begin{align}\label{idealstatefull}
u(t) = c([x(t);f_R(x(t),t)])
\end{align}
becomes \eqref{eq:real_closedfull} for all $t\geq0$. However, in order to find a data-driven construction of feedback controller \eqref{idealstatefull}, we need to identify the function $c$ and design a reference model $f_R$ without complete knowledge of $f$.
The following example shows that even if $f$ is not completely known, one can design a reference model $f_R$ such that 
Assumption \ref{ASM:new}-(b) holds.
 
\begin{exam}\label{exam:1}
Consider a system given in the Brunovsky canonical form \cite{Isidori-book}:
\begin{align*}
x^+ &= \begin{bmatrix} x_1^+ \\ x_2^+ \\ \vdots \\ x_{n-1}^+ \\ x_n^+ \end{bmatrix} 
= \begin{bmatrix} x_2 \\ x_3 \\ \vdots \\ x_n \\ a_n(x) \end{bmatrix} + \begin{bmatrix} 0 \\ 0 \\ \vdots \\ 0 \\ b_n(x) \end{bmatrix} u,
\end{align*}
where $x^+$, $x$, and $u$ imply $x(t+1)$, $x(t)$, and $u(t)$ respectively, $a_n:\R^n \to \R$ and $b_n: \R^n \to \R$ are unknown but smooth, and $u \in \R$.
Suppose that $b_n(x) \not = 0$ for all $x \in \R^n$.
Then by \eqref{eq:c}, Assumption \ref{ASM:new}-(a) holds with the (unknown) smooth function
\begin{equation}\label{eq:affine_c}
c([x;x^+]) = \frac{x_n^+ - a_n(x)}{b_n(x)}.
\end{equation}
Furthermore, \eqref{eq:image} (and thus, Assumption 2-(b)) is satisfied as long as the reference model has the form:
\begin{equation}\label{eq:f_R}
f_R(x,t) = \begin{bmatrix} x_2 \\ \vdots \\ x_n \\ f_n(x,t) \end{bmatrix},
\end{equation}
where $f_n(x,t)$ is an arbitrary smooth function. Although $a_n$ and $b_n$ are unknown, the smooth function $f_n(x,t)$ can be designed depending on the user's control goal.
For example, suppose that we want the state $x_1(t)$ to asymptotically track a given signal $\phi(t)$ by applying the control input $u=c([x;f_R(x,t)])$ of \eqref{idealstatefull}. Then, one can simply take 
\begin{equation}\label{eq:phi}
	f_n(x,t) = -\sum_{i=1}^n k_i (x_i(t) - \phi(t+i-1)) + \phi(t+n)
\end{equation}
in which, the state feedback gains $k_i$ are chosen such that the matrix
$$A = \begin{bmatrix} 0 & 1 & 0 & \cdots & 0 \\ 0 & 0 & 1 & \cdots & 0 \\ \vdots & \vdots & \vdots & \ddots & \vdots \\ -k_1 & -k_2 & -k_3 & \cdots & -k_n \end{bmatrix}$$
is Schur stable whose eigenvalues are located to meet the convergence performance.
(One can verify this using $\tilde x_i(t) := x_i(t) - \phi(t+i-1)$, $i=1,\ldots,n$, because they yield a stable linear system $\tilde x^+ = A \tilde x$.)
$\hfill\Box$
\end{exam}

Given a reference model that satisfies Assumptions \ref{asm:contractive} and \ref{ASM:new}-(b), let $\hat{c}:\R^{2n}\to\R^m$ be a data-driven identification of the function $c$. This means that the function $\hat{c}$ is obtained by using only available state/input measurements of system \eqref{SYSTEM:full} with some identification method on $c$. Before specifying the identification method, let us inspect a sufficient condition on the function $\hat{c}$ such that the closed-loop system \eqref{SYSTEM:full} with a controller
$u(t) = \hat{c}([x(t);f_R(x(t),t)])$
satisfies \eqref{eq:closeness}.
Let
$$B := \max_{x \in \XX, t \ge 0} \|(\partial f_R)/(\partial x)(x,t)\|$$
which is well-defined because $\XX$ is compact and $f_R$ is smooth in the first argument and uniformly bounded in the second argument.
Also, let $\lambda_{\min}$ and $\lambda_{\max}$ be the minimum and the maximum eigenvalues of the positive definite matrix $\Theta$ in Assumption \ref{asm:contractive}, respectively.
Finally, let $L_f$ be a Lipschitz constant of $f$ such that
$$\|f(x,u_a) - f(x,u_b)\| \le L_f \|u_a-u_b\|, \quad \forall x \in \XX, u_a, u_b \in \UU$$
which is well-defined due to smoothness of $f$ and compactness of $\XX$ and $\UU$.
Defining the set $\mathcal{C} := \mathcal{X}\times\mathcal{X}$, we present the following theorem.

\begin{thm}\label{THM1}
Under Assumptions \ref{asm:contractive} and \ref{ASM:new}, the closed-loop system \eqref{SYSTEM:full} with a controller $u(t) = \hat{c}([x(t);f_R(x(t),t)])$ satisfies
$$\|x(t) - x_R(t)\| < \epsilon, \qquad \forall t \ge 0$$
where $x(0) \in \XX_R$ and $x_R$ is the solution to the reference model \eqref{eq:real_closedfull} with $x_R(0)=x(0)$, if 
\begin{equation}\label{eq:regression_error}
	\| \hat{c}(\xi) - c(\xi) \| \le M, \qquad \forall \xi \in \mathcal{C}
\end{equation}
with $M>0$ such that
\begin{equation}\label{eq:M}
	\frac{2 L_f M \lambda_{\max} B \epsilon+ L_f^2 M^2 \lambda_{\max}}{\lambda_{\min}(1-\gamma)} < \epsilon^2 .
\end{equation}
\end{thm}

\begin{proof}
	We first note that, for each $x, x_R \in \XX$, there exists $ w \in \XX$ such that
	\begin{align*}
		&(f_R(x,t) - f_R(x_R,t))^T \Theta (f_R(x,t) - f_R(x_R,t)) \\
		&\le (x-x_R)^T \left( \frac{\partial f_R}{\partial x}(w,t) \right)^T \Theta \left( \frac{\partial f_R}{\partial x}(w,t) \right) (x-x_R)
	\end{align*}	
	which can be proved by the mean-value theorem with some trick\footnote{One can refer to \cite[Appendix A]{JWkim} for details.}.
	Now, let 
	$$V(x,x_R) = (x-x_R)^T \Theta (x-x_R).$$
	Then, since 
	$$x^+ = f(x,\hat{c}) = f_R(x,t) + [f(x,\hat{c})-f(x,c)],$$
	it follows, as long as $x \in \XX$ and $x_R \in \XX_R$, that
\begin{align}\label{thm:V}
	V^+ &= \left(f_R(x,t) - f_R(x_R,t) + [f(x,\hat{c})-f(x,c)] \right)^T \Theta \nonumber\\
	&\quad \times \left(f_R(x,t) - f_R(x_R,t) + [f(x,\hat{c})-f(x,c)] \right) \nonumber\\
	&\le (x-x_R)^T \left( \frac{\partial f_R}{\partial x}(w,t) \right)^T \Theta \left( \frac{\partial f_R}{\partial x}(w,t) \right) (x-x_R) \nonumber\\
	&\qquad +2L_fM \lambda_{\max} B \|x-x_R\| + L_f^2M^2 \lambda_{\max}\nonumber\\
	&\le \gamma V + 2L_fM \lambda_{\max} B \|x-x_R\| + L_f^2M^2 \lambda_{\max}.
\end{align}

We conclude the proof with mathematical induction.
Since $x(0)=x_R(0)\in\XX_R$ when $t=0$, the above inequality implies that
\begin{align*}
\|x(1)-x_R(1)\|^2 \le V(1)/\lambda_{\min} \le L_f^2M^2 \lambda_{\max}/\lambda_{\min} < \epsilon^2
\end{align*}
in which, the last inequality follows from \eqref{eq:M}.
Since $x_R(1) \in \XX_R$, it also follows that $x(1) \in \XX$ by the definition of $\XX$.
Now, suppose that $\|x(t)-x_R(t)\| < \epsilon$ for $t=0,1,\cdots,k-1$ (so that $x(t) \in \XX$ for $t=0,1,\dots,k-1$).
With 
$$\beta := 2 L_f M \lambda_{\max} B \epsilon + L_f^2 M^2 \lambda_{\max}$$
for convenience, we can repeatedly apply \eqref{thm:V} to obtain
\begin{align*}
V(k) &\le \gamma V(k-1) + \beta \le \gamma(\gamma V(k-2) + \beta) + \beta \\
&\le \quad \cdots \quad \le \gamma^{k} V(0) + \beta \frac{1-\gamma^{k}}{1-\gamma} \le \frac{\beta}{1-\gamma}.
\end{align*}
Then, we achieve with \eqref{eq:M} that
\begin{align*}
	\| x(k) - x_R(k) \|^2 &\le \frac{V(k)}{\lambda_{\min}}\le \frac{\beta}{\lambda_{\min}(1-\gamma)} < \epsilon^2.
\end{align*}
As a result, $\|x(k) - x_R(k)\| < \epsilon$ and $x(k) \in \XX$, which completes the proof.
\end{proof}

Theorem \ref{THM1} simply shows that if we have sufficiently accurate identification of the function $c$, it can directly be used for designing data-driven version of controller \eqref{idealstatefull} to guarantee sufficient performance of model reference control \eqref{eq:closeness}. In the next section, we utilize Gaussian process regression (GPR) as an identification method to obtain $\hat{c}$ and present how to collect state/input data to achieve sufficiently accurate identification result.

\section{Model Reference-Gaussian Process Regression}\label{sec:main}

In this section, we propose a controller generated by GPR using state/input data from system \eqref{SYSTEM:full}.
The resulting data-driven controller can generate control inputs that emulate those produced by \eqref{idealstatefull}.

We use GPR to identify the function $c$ in \eqref{idealstatefull} itself and to do so, we additionally assume that the function $c$ is infinitely differentiable in the set $\mathcal{C}.$ This approach involves treating $[x(t);x(t+1)]$ as input data and $u(t)$ as output data for the function $c$, based on the relation $u(t) = c([x(t);x(t+1)])$ from Assumption \ref{ASM:new}-(a).
For GPR, we split the function $c(\cdot)$ as $c(\cdot) =: [c_1(\cdot);c_2(\cdot);\cdots;c_m(\cdot)]$, where $c_i:\mathbb{R}^{2n} \to \mathbb{R}.$

To perform the proposal, we first collect state/input data{\footnote{The subscript $\dd$ is employed to represent the data gathered from the system during an experiment.}} of system \eqref{SYSTEM:full} as
\begin{align}\label{DATASET}
    \{(x_\dd(t),u_\dd(t))\}_{t=1}^{{N}}
\end{align}
where $N$ is the total number of state/input data.
Then we rearrange the data as the training input
\begin{align*}
    \xi_\dd(t) := [ {x_\dd}(t);{x_\dd}(t+1)] \in \mathbb{R}^{2n}
\end{align*}
and the training output
$$u_\dd {(t)}\in \mathbb{R}^{m},$$
yielding the training dataset:
\begin{align}\label{TD}
    {\mathcal{D}_N}:=\left\{ (\xi_\dd(t), u_\dd (t) )
    \right\}_{t=1}^{N-1}.
\end{align}
For the training dataset ${\mathcal{D}_N},$ we define the set of training input as 
\begin{align}\label{traininginputset}
\mathcal{I}_{{\mathcal{D}_N}}: = \left\{ \xi_\dd(t)\right\}_{t=1}^{N-1} \subset \mathbb{R}^{2n}.
\end{align}

Collecting sufficiently long state/input data as \eqref{TD} for an unstable system may present a challenge. In such cases, combining state/input data from each short experiment with varying initial conditions may be more practical. For further details, refer to \cite[Remark 1]{HHH22}.

A Gaussian process (GP) is uniquely characterized by a mean function $m_i: {\mathcal{C}} \to \mathbb{R}$ and a covariance function $k_i: {\mathcal{C}} \times {\mathcal{C}} \to \mathbb{R}$ for $i = 1,\cdots,m$.
To identify each function $c_i$ for $i = 1,\cdots,m$, we utilize the GP with the zero function for the mean function and a squared exponential (SE) kernel for the covariance function:
\begin{align}\label{SEkernel}
    k_i({\xi},{\xi}') = \sigma_{i}^2 \text{exp}\left( -\frac12 (\xi-\xi')^T H_i^{-1} (\xi-\xi') \right),
\end{align}
where $\sigma_{i}$ and $H_i = {\rm diag}(h_{i,1}^2, \ldots, h_{i,2n}^2)$ are hyperparameters for $i=1,\cdots,m$.
The hyperparameters are determined through marginal likelihood optimization, a technique based on Bayesian principles \cite[Chapter~5]{GPML}.

Given the training dataset $\mathcal{D}_N$ defined in \eqref{TD}, the GP produces posterior mean and variance functions for a test input $\xi \in \mathcal{C}$ as
\begin{align}
    \mu^i_{\mathcal{D}_N}(\xi) &:= \mathbf{k}_i^T ({\xi}){\mathbf{K}}_i^{-1} \mathbf{u}^i, \label{mu} \\
    \sigma^i_{\mathcal{D}_N}(\xi) &:= k_i ({\xi},{\xi}) - \mathbf{k}_i^T ({\xi}) {\mathbf{K}}_i^{-1} \mathbf{k}_i ({\xi}), \label{sig}
\end{align}
respectively, where 
\begin{align*}
\mathbf{u}^i &:= [{u}^i_\dd{(1)};\cdots;{u}^i_\dd{({N}-1)}],\\
\mathbf{k}_i ({\xi})&:= [k_i(\xi_\dd{(1)},{\xi});\cdots;k_i(\xi_\dd{({N}-1)},{\xi})],\\
\mathbf{K}_i &:= 
\resizebox{.42\textwidth}{!}{$
\begin{bmatrix}
k_i(\xi_\dd{(1)},\xi_\dd{(1)}) & \cdots & k_i(\xi_\dd{(1)},\xi_\dd{({N}-1)})\\
\vdots & \ddots & \vdots \\
k_i(\xi_\dd{({N}-1)},\xi_\dd{(1)}) & \cdots & k_i(\xi_\dd{({N}-1)},\xi_\dd{({N}-1)})
\end{bmatrix}
$}
\end{align*}
for $i=1,\cdots,m,$ where ${u}_\dd{(t)} =: [{u}^1_\dd{(t)};{u}^2_\dd{(t)};\cdots;{u}^m_\dd{(t)}].$
It is noted that each posterior mean function $\mu^i_{\mathcal{D}_N}$ is in fact the estimation of the function $c_i$ using only the state/input data of system \eqref{SYSTEM:full}. The confidence of this estimation is represented by the posterior variance function $\sigma^i_{\mathcal{D}_N}$.

Finally, with the reference model $f_R$, we construct the model reference GPR (MR-GPR) controller by combining posterior mean functions $\mu^i_{\mathcal{D}_N}$ to $\mu_{\mathcal{D}_N}$ as 
\begin{align}\label{controllerfull}
\begin{split}
u(t) &= \mu_{\mathcal{D}_N} ([x(t); f_R(x(t),t)])\\
     &= \begin{bmatrix}
        \mu^1_{\mathcal{D}_N} ([x(t); f_R(x(t),t)])\\
        \vdots\\
        \mu^m_{\mathcal{D}_N} ([x(t); f_R(x(t),t)])
    \end{bmatrix}
\end{split} 
\end{align}
which is a state feedback controller.
We also combine the posterior variance functions $\sigma^i_{\mathcal{D}_N}$ to $\sigma_{\mathcal{D}_N}$ as 
\begin{align}\label{variancefull}
\begin{split}
\sigma_{\mathcal{D}_N} ([x(t); f_R(x(t),t)]) = \begin{bmatrix}
        \sigma^1_{\mathcal{D}_N} ([x(t); f_R(x(t),t)])\\
        \vdots\\
        \sigma^m_{\mathcal{D}_N} ([x(t); f_R(x(t),t)])
    \end{bmatrix}.
\end{split}
\end{align}
The following lemma describes the difference between the ideal controller $c$ in \eqref{idealstatefull} and the MR-GPR controller $\mu_{\mathcal{D}_N}$ in \eqref{controllerfull} using the posterior variance function $\sigma_{\mathcal{D}_N}$ in \eqref{variancefull}.

\begin{lem}\label{LEMMA1}
\cite[Corollary~3.11]{Kanagawa18} Under Assumption \ref{ASM:new}, for the training dataset ${\mathcal{D}_N}$, it holds that
\begin{align*}
     ||\mu_{\mathcal{D}_N} ({\xi}) - c({\xi}) || \leq ||c||_k \cdot ||\sigma_{\mathcal{D}_N} ({\xi})||,~~~\forall {\xi} \in {\mathcal{C}},
\end{align*}
where $||c||_k: = \max (||c_1||_{k_1},\cdots,||c_m||_{k_m})$ and $||c_i||_{k_i}$ is the RKHS norm \cite{Kanagawa18} of the function $c_i$ under the kernel $k_i$.$\hfill\Box$
\end{lem}

Note that the upper bound of the identification error can be reduced arbitrarily by sufficiently reducing the posterior variance $\sigma_{\mathcal{D}_N} ({\xi})$.

Before going into the next lemma,
we define the set of the training inputs in the closed ball centered at $\xi$ with radius $\rho > 0$ as $\mathbb{B}_\rho (\xi) = \{\xi' \in \mathcal{I}_{\mathcal{D}_N} : ||\xi'-\xi|| \leq \rho \}$.
Also, we define a Lipschitz constant of the kernel $k_i$ as $L_{k_i}$.

\begin{lem}\label{LEMMA2}
\cite[Corollary~3.2]{Lederer19} Under Assumption \ref{ASM:new}, for the training dataset ${\mathcal{D}_N}$ as a function of $N$, if there exists a function $\rho({N})$ such that 
\begin{align*}
    0 < \rho ({N}) &\leq \min(\frac{k_1(\xi,\xi)}{L_{k_1}},\cdots,\frac{k_m(\xi,\xi)}{L_{k_m}}), ~~\forall {N} \in \mathbb{N}\\
    \lim_{{N} \to \infty} \rho ({N}) &= 0\\
    \lim_{{N} \to \infty} |\mathbb{B}_{\rho({N})} (\xi)| & = \infty
\end{align*}
for all $\xi \in \mathcal{C}$, then $||\sigma_{\mathcal{D}_N} ({\xi})|| \rightarrow 0$ as ${N} \rightarrow \infty$ for every $\xi \in \mathcal{C}$.
$\hfill\Box$
\end{lem}

The existence of the function $\rho$ may appear to be limited in Lemma \ref{LEMMA2}, but the method of constructing the dataset $\mathcal{D}_N$ to ensure the existence of $\rho$ is described in \cite{Lederer19}.

By Lemma \ref{LEMMA2}, we can arbitrarily reduce the gap between the ideal controller $c(\xi)$ and the MR-GPR controller $\mu_{\mathcal{D}_N}(\xi)$ by using a sufficiently large and dense training input set $\mathcal{I}_{{\mathcal{D}_N}}$, such that there always exist sufficiently many training inputs
in a ball centered at the point $\xi$ with a radius which approaches to $0$ as ${N} \rightarrow \infty$. Finally, combining Lemmas \ref{LEMMA1} and \ref{LEMMA2} with Theorem \ref{THM1}, we state the following theorem without proof.

%%%%%%%%%%%%%%%%%%%%%%%%%%%%%%%%%%%%%%%%%%%%%%%%%%%%%%%%%%%%%%%%%%%%%%%%%%%%%%%%
\begin{thm}
Under Assumptions  \ref{asm:contractive} and \ref{ASM:new},
suppose that the training dataset $\mathcal{D}_N$ satisfies the sufficient condition of Lemma \ref{LEMMA2}. Then, given $\epsilon>0$, there exists $\bar{N}\in\N$, such that for every $N>\Bar{N}$, the closed-loop system 
\eqref{SYSTEM:full} with the MR-GPR controller $u(t) = \mu_{\mathcal{D}_{N}} ([x(t);f_R(x(t),t)])$ sastisfies
$$\|x(t) -x_R(t)\| < \epsilon, ~~~~\forall t \ge 0$$
where $x(0) \in \XX_R$ and $x_R$ is the solution to the reference model \eqref{eq:real_closedfull} with $x_R(0)=x(0)$.
$\hfill\Box$
\end{thm}

\medskip

In practice, measurement noise is unavoidable.
While GPR can effectively handle the noise on the output of a function that is to be estimated, handling input noise is still under study (see, e.g., \cite{GPRinput1,GPRinput2}).
For instance, the authors of \cite{GPRinput1} utilize a Taylor expansion approximation to obtain a corrective term in the posterior mean and variance functions, while there is no explicit error bound estimation available for the scenario.

In this paper, we simply propose a modification of covariance in \eqref{mu} and \eqref{sig}, motivated by the following example.

\begin{exam}\label{example2}
Let us consider a linear system
\begin{equation}\label{eq:linsys}
    x(t+1) = Ax(t)+Bu(t),
\end{equation}
where $x\in\R^n$, $u\in\R^m$, and $B\in\R^{n\times m}$ is left invertible.  
It is clear that from \eqref{eq:c}, the function $c$ can be found as
\begin{equation}\label{eq:linc}
c([x(t);x(t+1)])=B^\dagger (x(t+1)-Ax(t)),    
\end{equation}
where $B^\dagger$ is a left-inverse of $B$.
Suppose that we have noisy measurements
$$
\tilde{x}(t) = x(t) + w_x(t),
$$
where $w_x$ is a white noise such that $w_x(t)\sim\mathcal{N}(0,\sigma_0^2I)$ for every time step $t$.
Then from \eqref{eq:linsys} and \eqref{eq:linc}, we obtain
\begin{align*}
    u(t) &= B^\dagger (x(t+1)-Ax(t)) \\
    &=B^\dagger (\tilde{x}(t+1)-A\tilde{x}(t)) - B^\dagger(w_x(t+1)-Aw_x(t))\\
    &= c([\tilde{x}(t);\tilde{x}(t+1)]) - B^\dagger(w_x(t+1)-Aw_x(t)).
\end{align*}
Defining $\tilde{u}(t):=u(t)+w_u(t)$, where
\begin{align*}
w_u(t):=B^\dagger(w_x(t+1)-Aw_x(t)),
\end{align*}
yields
\begin{equation*}\label{eq:noise}
    c([\tilde{x}(t);\tilde{x}(t+1)]) = \tilde{u}(t).
\end{equation*}
This indicates 
that the function $c$ can be identified by using the training input 
\begin{align*}
    \tilde{\xi}(t) := [ \tilde{x}(t);\tilde{x}(t+1)] \in \mathbb{R}^{2n}
\end{align*}
which can be considered as a noise-free data
and the corresponding (noisy) training output 
$$
\tilde{u}(t)=u(t)+w_u(t)\in\mathbb{R}^m,
$$
where
$w_u(t)\sim \mathcal{N}(0, \sigma_0^2B^\dagger (I-AA^\top)(B^\dagger)^\top)$.
$\hfill\Box$
\end{exam}

Motivated by the example, we modify the covariance matrix of \eqref{mu} and \eqref{sig} as if there are output noise like:
\begin{align*}
        \mu^i_{\mathcal{D}_N}(\xi) &:= \mathbf{k}_i^T ({\xi})({\mathbf{K}}_i+{\sigma_{i,n}}I)^{-1} \mathbf{u}^i,  \\
    \sigma^i_{\mathcal{D}_N}(\xi) &:= k_i ({\xi},{\xi}) - \mathbf{k}_i^T ({\xi}) ({\mathbf{K}}_i+{\sigma_{i,n}}I)^{-1} \mathbf{k}_i ({\xi}),
\end{align*}
where ${\sigma_{i,n}}$ is the additional hyperparameter.
We demonstrate its effect in the next section.

\section{Illustrative Example}\label{sec:ex}
In this section, an illustrative example is presented to describe the utility of the proposed data-driven controller.
Consider an inverted pendulum system discretized by Euler method as
\begin{equation}\label{ex}
\begin{split}
    z_1 (t+1) &= z_2(t),\\
    z_2 (t+1) &= a(z(t)) + \frac{T^2}{ml^2}u(t),
\end{split}
\end{equation}
where $a(z) := z_2 + \frac{gT^2}{l}\sin(z_1)+(1-\frac{\mu T}{ml^2})(z_2-z_1)$, and $g$ is gravitational constant, $l$ is the distance from the base to the center of mass of the balanced body, $\mu$ is the coefficient of rotational friction, $m$ is the mass to be balanced, and $T$ is a sampling period.
Following Example \ref{exam:1}, it is obvious that
\begin{equation}\label{eq:c_examp}
c([z(t);z(t+1)]) = \frac{ml^2}{T^2}(z_2(t+1) - a(z(t)) )
\end{equation}
by \eqref{eq:affine_c}.
With a given reference signal $\phi (t)$, the reference model $f_R$ is chosen following \eqref{eq:f_R}, namely,
$$
f_R(z(t),t)= 
\begin{bmatrix}
    z_2(t) \\ \phi(t+2)
\end{bmatrix}
$$
and it also satisfies Assumption \ref{asm:contractive}.
Now, we design MR-GPR controller
\begin{align*}
    u(t) = \mu_{\mathcal{D}_N} ([z(t);f_R(z(t),t)]),
\end{align*}
which is expected to make $z_1(t)$ track the signal $\phi(t)$ under system parameters $g=9.8, m=l=0.2,\mu=0.01$, and sampling period $T=0.1$.

Let us collect the data for identifying function $c$ in \eqref{eq:c_examp}. Instead of using only one trajectory of input and state as \eqref{DATASET} and \eqref{TD}, we use the following training dataset
$$
\mathcal{D}_N = \left\{ (\xi_\dd^i(0), u^i_\dd (0) )
    \right\}_{i=1}^{N-1},
$$
where
\begin{align*}
    \xi_\dd^i(0) = [ {z_\dd}^i(0);{z_\dd}^i(1)]
\end{align*}
is collected by the $i$-th experiment with a random initial condition ${z_{1,\dd}}^i(0)\in[-\pi,\pi]$, ${z_{2,\dd}}^i(0) = 0,$ and a random input $ u^i_\dd (0)\in[-1,1]$ for $i=1,\ldots,N-1$.
In this simulation, we adopted an alternative definition for $\mathcal{D}_N$ which differs from the definition in \eqref{TD}, but can be justified by the approach presented in \cite[Remark 1]{HHH22}. The hyperparameters of the SE kernel in \eqref{SEkernel} were determined by optimizing the marginal likelihood through the GPML toolbox \cite{gpmltoolbox} using the training data ${\mathcal{D}_N}$. 

The performance of the MR-GPR controller is compared to that of the ideal controller in the closed-loop system with two different initial conditions (in all cases, we set the initial conditions of each system to $(z_1(0),z_2(0))\in\{(-\pi,0), (\pi/2,0)\}$), where the MR-GPR controller is trained using $N=300$ experiments. As shown in Fig.~\ref{FIG:toy_epi300}, the performance of the MR-GPR controller pretty well follows that of the ideal controller.
Additionally, the MR-GPR controller is trained with state data in the presence of additive white Gaussian noise with a signal-to-noise ratio of 20 dB for $N=3000$ experiments. The result is illustrated in Fig.~\ref{FIG:toy_epi300_input_noise}.
The figure reveals that the control performance is less-effective under measurement noise on the state (which plays the role of input noise for GPR). Therefore, we implement the method discussed at the end of the previous section.
By setting the hyperparameter $\sigma_{i,n} = 0.5,$ the result is shown in Fig.~\ref{FIG:toy_epi300_output_noise}. 
It is observed that, in the presence of noise in the state, the proposed MR-GPR controller obtains better result.

\begin{figure}[t]
\centering
\includegraphics[width=\columnwidth]{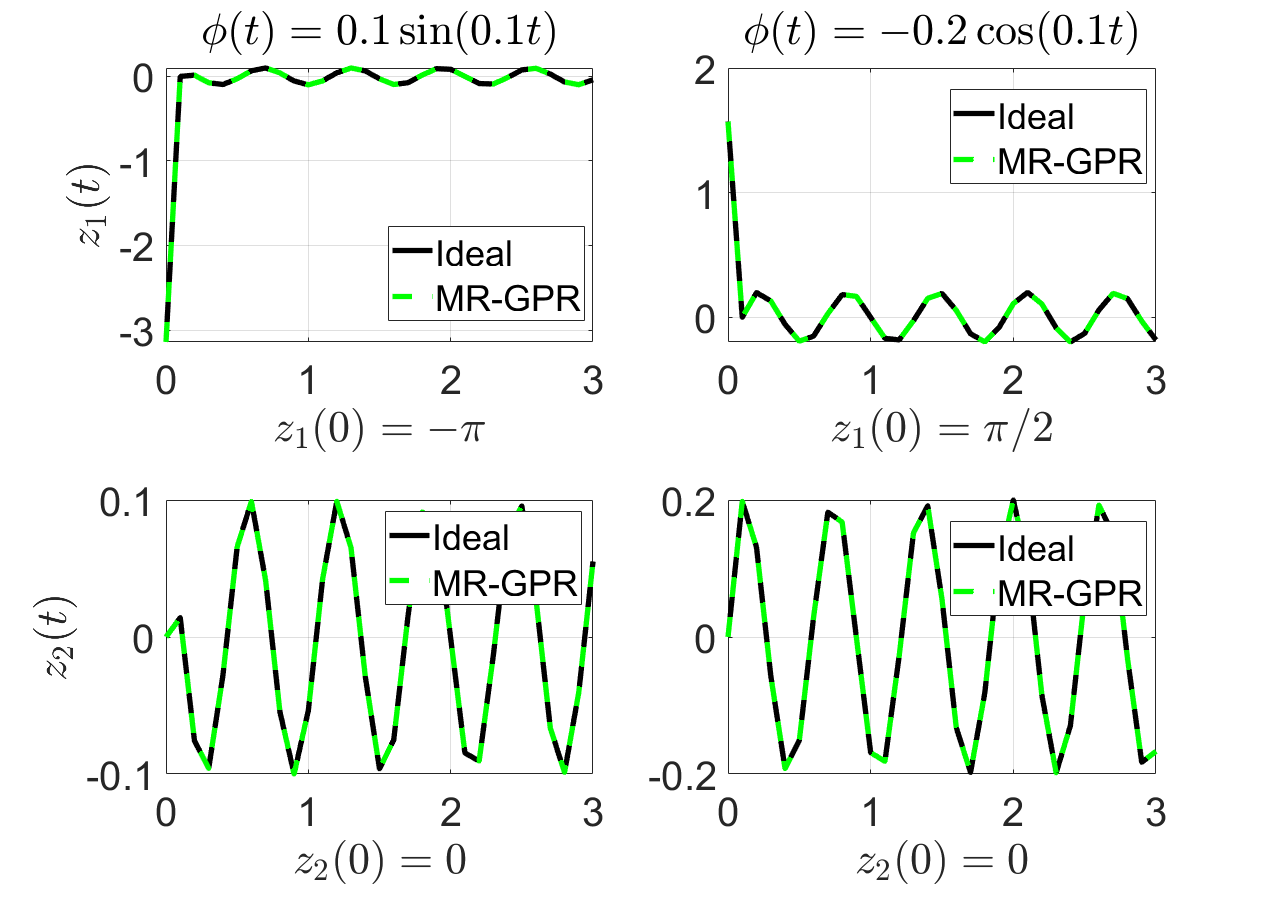}
\caption{State trajectories of the system \eqref{ex} with ideal 
controller $c$ (black line) and MR-GPR controller $\mu_{\mathcal{D}_N}$ (green dashed line) designed by the data of $N=300$ experiments from different initial conditions and different $\phi(t)$.\label{FIG:toy_epi300}}
\end{figure}
\begin{figure}[t]
\centering
\includegraphics[width=\columnwidth]{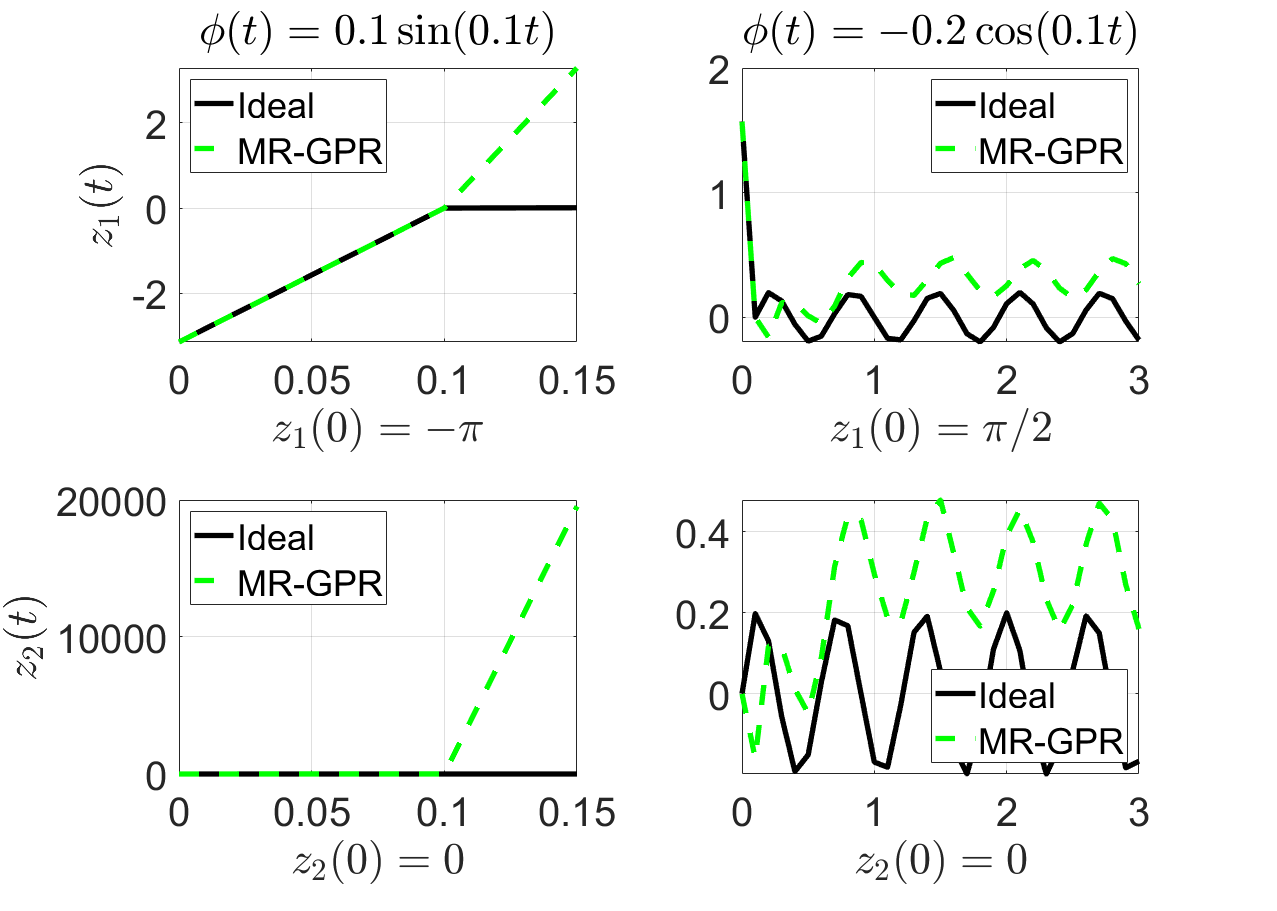}
\caption{State trajectories of the system \eqref{ex} with ideal 
controller $c$ (black line) and MR-GPR controller $\mu_{\mathcal{D}_N}$ (green dashed line) designed by the data of $N=3000$ experiments under the state noise in the state $z(t)$ from different initial conditions. We only draw $0.15s$ for the left figures because they show unstable system. \label{FIG:toy_epi300_input_noise}}
\end{figure}
\begin{figure}[t]
\centering
\includegraphics[width=\columnwidth]{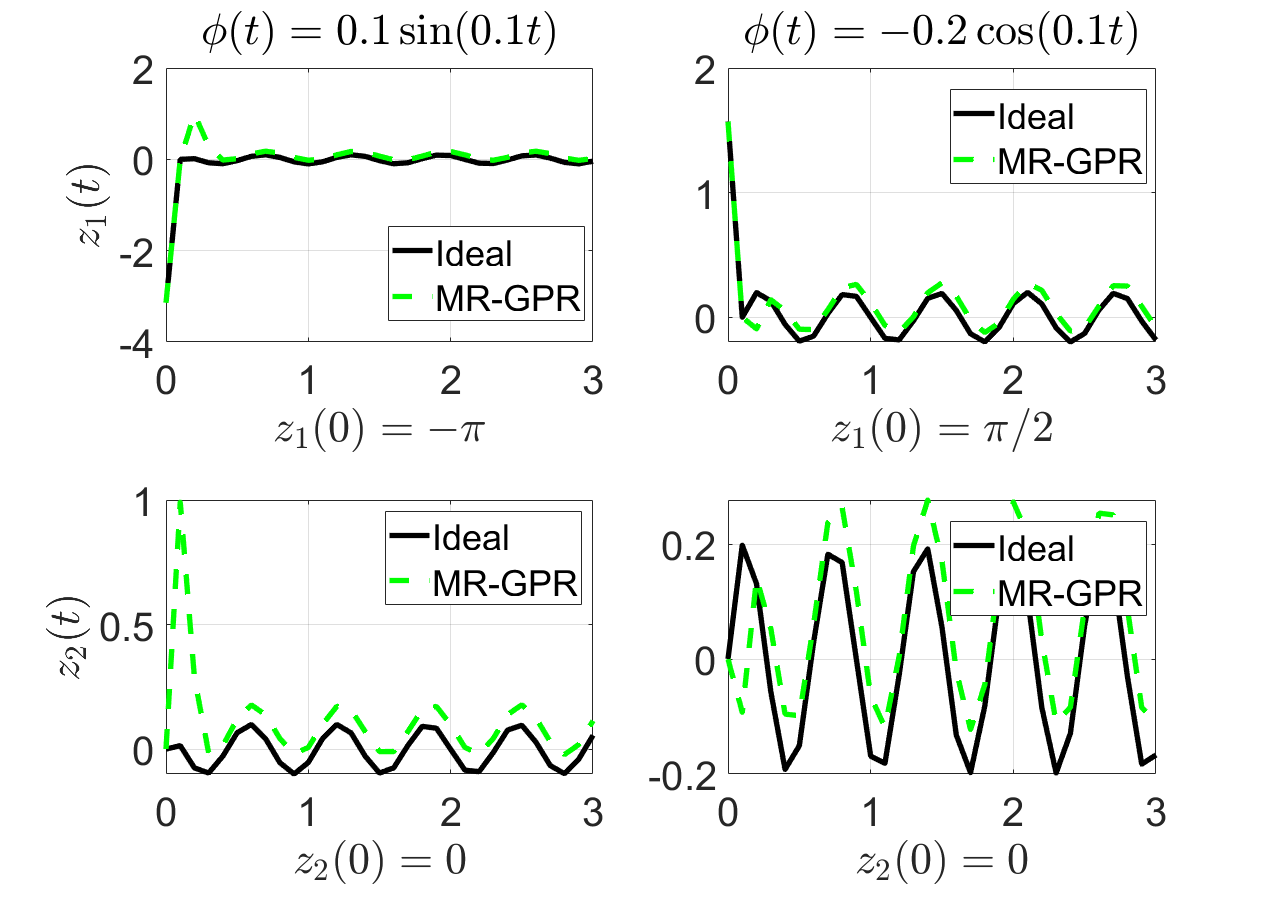}
\caption{State trajectories of the system \eqref{ex} with ideal 
controller $c$ (black line) and MR-GPR controller $\mu_{\mathcal{D}_N}$ (green dashed line) designed by the data of $N=3000$ experiments under the state noise treating it as input noise.\label{FIG:toy_epi300_output_noise}}
\end{figure}

%%%%%%%%%%%%%%%%%%%%%%%%%%%%%%%%%%%%%%%%%%%%%%%%%%%%%%%%%%%%%%%%%%%%%%%%%%%%%%%%
\section{Conclusion}\label{sec:conc}

In this study, we introduced a data-driven state feedback controller, referred to as the MR-GPR controller, for nonlinear discrete-time systems under some assumptions. The proposed controller was designed based on the GPR, trained solely on the state/input data of the system, without requiring any prior knowledge of the system's underlying physics or mathematical model. Additionally, the usefulness of the proposed controller was verified through numerical experiments.

%%%%%%%%%%%%%%%%%%%%%%%%%%%%%%%%%%%%%%%%%%%%%%%%%%%%%%%%%%%%%%%%%%%%%%%%%%%%%%%%


\begin{thebibliography}{99}

\bibitem{GPML}
C.~K.~I.~Williams and C.~E.~Rasmussen,
\emph{Gaussian processes for machine learning}, Cambridge, MA: MIT Press, 2006.

\bibitem{KocijanBook}
J.~Kocijan, 
\emph{Modelling and control of dynamic systems using Gaussian process models}, Cham: Springer International Publishing, 2016.

\bibitem{RDCA03}
R.~Murray-Smith, D.~Sbarbaro, C.~E.~Rasmussen, and A.~Girard, 
``Adaptive, cautious, predictive control with Gaussian process priors,'' \emph{IFAC Proceedings}, vol.~36, no.~16, pp.~1155-1160, 2003.

\bibitem{JRCA04}
J.~Kocijan, R.~Murray-Smith, C.~E.~Rasmussen, and A.~Girard,
``Gaussian process model based predictive control,'' in \emph{Proceedings of American Control Conference}, pp.~2214-2219, 2004.

\bibitem{GEF16}
G.~Cao, E.~M.~K.~Lai, and F.~Alam, ``Gaussian process model predictive control of an unmanned quadrotor,'' \emph{Journal of Intelligent \& Robotic Systems}, vol.~88, no.~1, pp.~147-162, 2017.

\bibitem{CATJ16}
C.~J.~Ostafew, A.~P.~Schoellig, T.~D.~Barfoot, and J.~Collier, ``Learning-based nonlinear model predictive control to improve vision-based mobile robot path tracking,'' \emph{Journal of Field Robotics}, vol.~33, no.~1, pp.~133-152, 2016.

\bibitem{LJM19}
L.~Hewing, J.~Kabzan, and M.~N.~Zeilinger, ``Cautious model predictive control using Gaussian process regression,'' \emph{IEEE Transactions on Control Systems Technology}, vol.~28, no.~6, pp.~2736-2743, 2019.

\bibitem{JLAM19}
J.~Kabzan, L.~Hewing, A.~Liniger, and M.~N.~Zeilinger, ``Learning-based
model predictive control for autonomous racing,'' \emph{IEEE Robotics and Automation Letters}, vol.~4, no.~4, pp.~3363-3370, 2019.


\bibitem{Scaramuzza}
G.~Torrente, E.~Kaufmann, P.~Föhn, and D.~Scaramuzza, ``Data-driven MPC for quadrotors,'' \emph{IEEE Robotics and Automation Letters}, vol.~6, no.~2, pp.~3769-3776, 2021.

\bibitem{C05}
J.~J.~Craig, 
\emph{Introduction to Robotics: Mechanics and Control}, Pearson Education International, Upper Saddle River, NJ, 2005.


\bibitem{NSP08}
D.~Nguyen-Tuong, M.~Seeger, and J.~Peters, 
``Computed torque control with nonparametric regression models,'' in \emph{Proceedings of American Control Conference}, pp.~1-6, 2008.

\bibitem{NSP09}
D.~Nguyen-Tuong, M.~Seeger, and J.~Peters, 
``Model learning with local Gaussian process regression,'' \emph{Advanced Robotics}, vol.~23, no.~15, pp.~2015-2034, 2009.

\bibitem{NPSS08}
D.~Nguyen-Tuong, J.~Peters, M.~Seeger, and B.~Schölkopf, 
``Learning inverse dynamics: a comparison,'' in \emph{Proceedings of the European Symposium on Artificial Neural Networks}, pp.~13-18, 2008.

\bibitem{N01}
M.~Nørgaard, O.~Rvn, N.~K.~Poulse, and L.~K.~Hansen, 
\emph{Neural networks for modelling and control of dynamic systems}, Springer-Verlag London Limited, London, England, 2001.

\bibitem{HHH22}
H.~Kim, H.~Chang, and H.~Shim, ``Model Reference Gaussian Process Regression: Data-Driven Output Feedback Controller,'' 2022, Accepted in ACC 2023. \url{http://arxiv.org/abs/2210.02494}



\bibitem{Lohmiller98}
W.~Lohmiller and J.~Slotine,
``On contraction analysis for non-linear systems,''
\emph{Automatica} 34.6 (1998): 683-696.


\bibitem{Isidori-book}
A.~Isidori, 
{\it Nonlinear control systems: an introduction},
Berlin, Heidelberg: Springer Berlin Heidelberg, 1985.

\bibitem{JWkim}
J.~W.~Kim, J.~K.~Lee, D.~Lee, and H.~Shim, ``A DesignMethod of Distributed Algorithms via Discrete-time Blended Dynamics Theorem,'' 2022. \url{http://arxiv.org/abs/2210.05142}


\bibitem{Kanagawa18}
M.~Kanagawa, P.~Hennig, D.~Sejdinovic, and B.~K.~Sriperumbudur, ``Gaussian processes and kernel methods: a review on connections and equivalences,'' 2018. \url{http://arxiv.org/abs/1807.02582}

\bibitem{Lederer19}
A.~Lederer, J.~Umlauft, and S.~Hirche, ``Posterior variance analysis of Gaussian processes with application to average learning curves,'' 2019. \url{http://arxiv.org/abs/1906.01404}




\bibitem{GPRinput1}
A.~McHutchon and C.~Rasmussen, 
``Gaussian process training with input noise,''
\emph{Advances in neural information processing systems}, 24, 2011.

\bibitem{GPRinput2}
H.~Bijl, T.~B.~Schön, J.~W.~van~Wingerden, and M.~Verhaegen,
``System identification through online sparse Gaussian process regression with input noise,''
\emph{IFAC Journal of Systems and Control}, vol.~2, pp.~1-11, 2017.


\bibitem{gpmltoolbox}
C.~E.~Rasmussen and H.~Nickisch,
``The GPML toolbox version 4.0,'' 
\emph{Technical Documentation}, 2016.



\end{thebibliography}
\end{document}